\documentclass[11pt]{article}

\usepackage{amsmath}
\usepackage{amssymb}
\usepackage{amsthm}

\theoremstyle{plain}
\newtheorem{theorem}{Theorem}
\newtheorem{lemma}[theorem]{Lemma}
\newtheorem{proposition}[theorem]{Proposition}
\newtheorem{corollary}[theorem]{Corollary}

\theoremstyle{remark}
\newtheorem{remark}{Remark}

\newcommand{\Cpx}{\mathbb{C}}
\newcommand{\Proj}{\mathbb{P}}
\newcommand{\ip}[2]{\langle #1 , #2 \rangle}
\newcommand{\nm}[1]{\lVert #1 \rVert}
\newcommand{\ab}[1]{\lvert #1 \rvert}
\newcommand{\ta}{\theta_{a}}            % true (actual) parameter
\newcommand{\Ra}{R_{a}}
\newcommand{\ua}{u_{a}}
\newcommand{\lmin}{\lambda_{\min}}
\newcommand{\lmax}{\lambda_{\max}}
\newcommand{\hrm}{\mathsf{H}}           % Hermitian transpose
\newcommand{\trn}{\mathsf{T}}           % transpose
\DeclareMathOperator{\Rea}{Re}
\DeclareMathOperator{\Ima}{Im}
\DeclareMathOperator{\dis}{dist}
\DeclareMathOperator{\diag}{diag}

\allowdisplaybreaks

\begin{document}

\title{\bfseries Unimodality and Radial Monotonicity of the Magnetic
Resonance Fingerprinting $T_1/T_2$ Matching Objective}

\author{Ze Wang\\[2pt]
\normalsize Department of Diagnostic Radiology and Nuclear Medicine,\\
\normalsize University of Maryland School of Medicine,
Baltimore, MD, USA}

\date{}
\maketitle

%---------------------------- ABSTRACT ---------------------------------
\begin{abstract}
\noindent
Magnetic resonance fingerprinting (MRF) estimates tissue parameters by
matching an acquired MR signal time course to entries in a Bloch- or
EPG-simulated dictionary. However, no study has yet proven the
uniqueness of the matching results. In two earlier studies by
exhaustive objective mapping I showed that for two widely used MRF sequences the
normalized-correlation objective exhibits a single dominant peak at the
true $T_1/T_2$ values and decreases smoothly away from that peak. These
empirical properties motivated the fast MRF-ZOOM search algorithm even without using a pre-generated signal dictionary, but
their theoretical basis has remained incomplete. The purpose of this work is to develop a
mathematical framework for the MRF matching objective under
normalized-correlation matching. I first show that, for any fixed
radiofrequency (RF) flip angle (FA) and repetition-time (TR) schedule, each MRF signal sample
can be described as an exponential polynomial in the relaxation rates $R_1 = 1/T_1$ and
$R_2 = 1/T_2$. When the TRs are commensurate, the signal further reduces
to a bivariate polynomial in $(e^{-\Delta R_1}, e^{-\Delta R_2})$. This
structure implies that the squared matching objective is real analytic
and excludes open plateaus. A projective-space formulation then
identifies the objective with the cosine of the Fubini--Study distance
between normalized fingerprints. Its second-order expansion yields a
curvature matrix equal, up to scale, to the Fisher information after
profiling out complex proton density; positive definiteness is therefore
equivalent to local identifiability, a finite local Cram\'{e}r--Rao
bound, and a strict non-degenerate local peak. I then identify an explicit
neighborhood of strict concavity and local radial monotonicity, and for
a specified sequence and target I provide finite Lipschitz certificates
for uniqueness of the global maximizer and, more strongly, for global
radial monotonicity and exclusion of secondary stationary points.
The analysis is
developed here for matching in the full signal space; the corresponding
geometry of subspace-compressed matching is treated in a companion
paper.
\end{abstract}


%========================================================================
\section{Introduction}

Magnetic resonance fingerprinting (MRF) \cite{Ma2013,Jiang2015}
estimates multiple tissue parameters from a single time-resolved
acquisition by encoding parameter-dependent signal evolutions through a
prescribed sequence of flip angles (FAs), repetition times (TRs), and,
where applicable, RF phases. The measured time course---the
fingerprint---is compared with simulated dictionary entries, and the
parameters associated with the best match are retained. In the standard
formulation considered here, similarity is measured by the modulus of
the normalized inner product, referred to below as the correlation
coefficient (CC).

The volume of the dictionary and the computational burden of exhaustive dictionary search grows rapidly
with the number of encoded parameters, their ranges and resolutions, and
the fingerprint length. This motivated the MRF-ZOOM family of algorithms
\cite{Wang2018,Wang2020}, which replaces exhaustive comparison with a
coarse-to-fine, parameter-separable search. Its efficiency relies on
empirical regularity of the matching objective. In \cite{Wang2018},
dense objective mapping suggested three properties: (I) off-resonance is
approximately orthogonal to $T_1$ and $T_2$ and the objective is
pseudo-periodic in off-resonance with approximate period $1/\mathrm{mean}(TR)$;
(II) the objective over $(T_1,T_2)$ is smooth and has a single peak at
the true parameters; and (III) $T_1$ and $T_2$ are approximately
orthogonal near that peak. Property (II) was later observed again for
the unbalanced SSFP (FISP) sequence \cite{Wang2020}, including after
dictionary compression.

Those heuristic observations were nevertheless supported primarily by finite
numerical experiments and a local Taylor argument. Dense mapping can
demonstrate the objective shape for a particular sequence and sampled
parameter grid, but it cannot establish a general result or exclude an
unsampled secondary peak. Likewise, the first-order argument in
\cite{Wang2018} explains the local sensitivity of the signal to $T_1$
offsets but does not establish global uniqueness and does not fully
explain the asymmetric flattening observed at long $T_1$. These
limitations motivate a sequence-level mathematical analysis of the
matching objective.

The present work separates the empirical observation that the MRF objective is ``single-peaked, smooth, and monotonic'' into several precise mathematical questions: Is the objective smooth? Why does the ground-truth parameter maximize it? When is this maximizer locally strict? Under what conditions is it globally unique? And when does the objective decrease monotonically along rays emanating from the truth? Crucially, these properties do not hold with equal generality. Smoothness
follows from the analytic structure of the MR signal recursion described by the Bloch equation \cite{Bloch1946} or Extended Phase Graph (EPG)
\cite{Weigel2015}, whereas
local strictness depends on local identifiability. Global uniqueness and
global radial monotonicity are stronger, sequence- and target-dependent
properties. My goal is therefore not to assert universal unimodality,
but to identify the structural results that hold generally and to derive
finite, verifiable conditions under which the stronger geometric
properties hold for a specified MRF sequence.

Throughout, matching is performed in the full $N$-dimensional signal
space. Many contemporary MRF pipelines instead match after projection
onto a low-dimensional temporal subspace
\cite{McGivney2014,Asslander2018}. Accordingly, I extend the work presented in this paper to the compressed dictionary space in a companion paper \cite{PaperII}, addressing which of the results presented in this paper survive
such a fixed linear compression, how compression perturbs the
objective, and how the subspace dimension should be chosen.

%========================================================================
\section{Problem setup and notation}
\label{sec:setup}

An MRF sequence \cite{Ma2013} consists of $N$ repetitions, each comprising an RF excitation followed by data acquisition. The $n$-th excitation has a FA $\alpha_n$ and phase $\varphi_n$, and the corresponding repetition time is $TR_n$; all sequence parameters are fixed and known. For simplicity, I focus on the parameter
vector $\theta = (T_1,T_2)$, defined on a compact rectangle
$\Theta \subset (0,\infty)^2$; off-resonance $df$ is treated separately
and is not considered here. The results established for
$T_1/T_2$ should not be assumed to extend automatically to additional
Parameters as their encoding geometries may differ. Throughout, we write
\begin{equation}\label{eq:rates}
R_1 = \frac{1}{T_1}, \qquad
R_2 = \frac{1}{T_2}, \qquad
E_{1j} = \exp(-TR_j\,R_1), \qquad
E_{2j} = \exp(-TR_j\,R_2),
\end{equation}
and reserve $R = (R_1, R_2)$ for the reciprocal, or rate, coordinates.
The Bloch equation \cite{Bloch1946} or EPG
\cite{Weigel2015} simulated dictionary entry is the complex vector
\begin{equation}\label{eq:xu}
x(\theta) = \bigl( s_1(\theta), \dots, s_N(\theta) \bigr) \in \Cpx^{N},
\qquad
u(\theta) = \frac{x(\theta)}{\nm{x(\theta)}},
\end{equation}
where $s_j$ is the transverse magnetization sampled at the $j$-th echo.
The measured fingerprint can be simplified as $y = \rho\,x(\ta) + n$,
with $\ta$ the true parameters, $\rho \in \Cpx$ an unknown complex
proton density absorbing both the spin density and the receive phase,
and $n$ complex Gaussian noise. Following
\cite{Ma2013,Wang2018,Wang2020}, matching is scored by
\begin{equation}\label{eq:C}
C(\theta) = \bigl\lvert \ip{\ua}{u(\theta)} \bigr\rvert,
\qquad \ua := u(\ta),
\end{equation}
with $\ip{p}{q} = \sum_j \overline{p}_j q_j$ the Hermitian inner
product; the modulus removes the unknown phase of $\rho$ and the
normalization removes its magnitude. MRF dictionary searching maximizes
$C$ over $\Theta$. The CC-targeted MRF optimization can be equally described through the following loss
\begin{equation}\label{eq:L}
L(\theta) = 1 - C(\theta)^2 = \nm{\Pi_{u(\theta)}\,\ua}^2,
\qquad
\Pi_{v} := I - v v^{\hrm}.
\end{equation}
Equation~\eqref{eq:L} is the squared residual after optimizing over an
arbitrary complex scale factor. $\Pi$ means the orthogonal
projector onto the complex orthogonal complement of the current signal,
and $\nm{\cdot}$ means the
Euclidean or induced operator norm as appropriate. Unless stated
otherwise, all local statements are made for a fixed target $\ta$ in the
interior of $\Theta$.

\begin{remark}
By the Cauchy--Schwarz inequality, $0 \le C(\theta) \le 1$, with equality at the upper
bound if and only if $u(\theta)$ and $\ua$ are complex-collinear. At the solution, $C(\ta) = 1$, meaning that the true parameters always attain the largest possible
objective value. The substantive ``single-peak'' question is therefore
whether any other parameter value also attains this global maximum or constitutes a distinct local maximum.
\end{remark}

%========================================================================
\section{The structure of the MRF signal manifold}
\label{sec:manifold}

The remainder of the analysis relies on a structural property of the
MR signal recursion as can be modeled through Bloch/EPG: although the fingerprint depends nonlinearly on
$T_1$ and $T_2$, that dependence lies within a highly restricted class of functions.

\begin{lemma}[affine relaxation dependent magnetization recursion]
\label{lem:affine}
Let $m_j$ denote the magnetization state after the $j$-th TR: the vector
$(M_x, M_y, M_z)$ in the Bloch description, or the stacked configuration
states $(F^{+}_{k}, F^{-}_{k}, Z_k)_k$ in the EPG description. Then
\begin{equation}\label{eq:recursion}
m_j = D_j(\theta)\,A_j\,m_{j-1} + (1 - E_{1j})\,e,
\end{equation}
where $A_j$ collects the RF rotation, the off-resonance rotation and, in
the EPG case, the gradient shift, and depends only on the sequence and
not on $\theta$; $D_j(\theta)$ is diagonal with entries drawn from
$\{E_{1j}, E_{2j}\}$; and $e$ is a fixed vector selecting the
longitudinal component of order zero. Note that the observed signal is
the transverse component of the magnetization state, i.e.\
$s_j = P\,m_j$, where $P$ represents the linear projection onto the
measured complex transverse magnetization.
\end{lemma}

\begin{proof}
In the Bloch description this is exactly the update of equations
(A2)--(A6) of \cite{Wang2018}: relaxation over one TR is the diagonal
map $\diag(E_{2j}, E_{2j}, E_{1j})$ followed by the addition of
$(1 - E_{1j})$ along $z$, while the RF pulse and the off-resonance
precession are rotations, which do not involve $T_1$ or $T_2$. In the
EPG description the RF mixing operator $T(\alpha_j, \varphi_j)$ and the
gradient shift operator act linearly on the configuration states with
sequence-determined coefficients only; relaxation multiplies every
transverse configuration by $E_{2j}$ and every longitudinal
configuration by $E_{1j}$, and regrowth adds $(1 - E_{1j})$ to the
zero-order longitudinal state alone. In both cases the
$\theta$-dependence is confined to the diagonal factor and to the scalar
offset, which is the assertion.
\end{proof}

\begin{proposition}[exponential-polynomial construct of MR signal recursion]
\label{prop:expoly}
For every $j$ there exist a finite number of complex coefficients
$c_{jk}$, determined by the sequence alone and the fixed initial magnetization state, and non-negative exponents
$a_{jk}$, $b_{jk}$, each equal to a sum of TRs, such that
\begin{equation}\label{eq:sj}
s_j(R_1, R_2)
 = \sum_k c_{jk}\, \exp\bigl( -a_{jk} R_1 - b_{jk} R_2 \bigr).
\end{equation}
In particular each $s_j$ extends to an entire function of
$(R_1, R_2) \in \Cpx^{2}$.
\end{proposition}

\begin{proof}
Unrolling the recursion \eqref{eq:recursion} gives
\[
m_n = \Bigl( \prod_{j=n}^{1} D_j A_j \Bigr) m_0
    + \sum_{i=1}^{n} \Bigl( \prod_{j=n}^{i+1} D_j A_j \Bigr)
      (1 - E_{1i})\, e .
\]
Expanding the matrix products, each resulting term is a product of one
diagonal entry from each $D_j$, times a constant built from the $A_j$.
Each such diagonal entry is either $E_{1j} = \exp(-TR_j R_1)$ or
$E_{2j} = \exp(-TR_j R_2)$, so their product is $\exp(-a R_1 - b R_2)$
with $a$ and $b$ sums of TRs. The scalar factor $(1 - E_{1i})$
contributes the two terms $1$ and $-\exp(-TR_i R_1)$, of the same form.
Since the number of terms is finite, the sum is a finite exponential
sum, and finite sums of exponentials of linear forms are entire (complex-differentiable (holomorphic) everywhere in its complex domain).
\end{proof}

Indeed, proposition~\ref{prop:expoly} provides the structural foundation for the
analysis in this paper. It tells that, for any fixed MRF schedule
(considering $T_1$ and $T_2$ here only), each measured signal sample
belongs to a finite exponential-polynomial family in the relaxation
rates $R_1$ and $R_2$. No regularity of the FA, RF-phase, or TR schedule
is required. Pseudo-randomization therefore does not conflict with the
analysis: it changes the coefficients and exponents of the resulting
exponential polynomial, but not its analyticity. The role of sequence
design is different. Through distinguishable MR signal time course perturbations induced by the incoherent $T_1$ and $T_2$ relaxations, an informative MR excitation and relaxation
schedule promotes local identifiability; this geometric condition is
formalized later through the information matrix $G$. The structural
results of Section~\ref{sec:smooth}, however, require only the Bloch/EPG
recursion and hold independently of how the schedule was designed. Coefficients a and b in Eq.~\ref{eq:sj} determine how much relaxation-history complexity can accumulate by the time $s_j$ is generated. Their sum can be called the degree.

\begin{corollary}[polynomial structure for commensurate TRs]
\label{cor:commensurate}
Suppose all sequence timings share a common time base: $TR_j = n_j \Delta$
with $n_j$ positive integers, and likewise for $TE$ and any preparation
interval. This holds whenever scanner timing is quantized to a finite
raster, and places no restriction on pseudo-random variation of the $n_j$;
one may always take $\Delta = \gcd_j TR_j$ ($\gcd$ means greatest common divisor). Define $q_1 = \exp(-\Delta R_1)$
and $q_2 = \exp(-\Delta R_2)$. Then each signal sample $s_j$ lies in
$\Cpx[q_1,q_2]$, with coefficients determined by the FA and RF-phase
schedule alone, and total degree at most $N_j = \sum_{l \le j} n_l$.
For equal TRs ($n_j \equiv n$) one may re-base to $\Delta = TR$, giving
$(E_1,E_2) = (q_1^{n}, q_2^{n})$ and total degree at most $j$ in
$(E_1,E_2)$.

\emph{Attainment.} The bound $N_j$ counts every interval up to and
including $TR_j$, whereas the $j$-th sample is acquired at an echo time
$TE = n_{TE}\Delta$ after the $j$-th excitation, so that only
$TR_1,\dots,TR_{j-1}$ and $TE$ have elapsed. Hence
\begin{equation}\label{eq:exactdeg}
\deg s_j \;\le\; d_j := \sum_{l<j} n_l + n_{TE} \;\le\; N_j ,
\end{equation}
and the second inequality is strict unless $n_{TE} = n_j$. In particular,
for equal TRs sampled at $TE = 0$ the attainable degree is $j-1$, not $j$.
Direct evaluation of the recursion in polynomial arithmetic confirms that
the first inequality in \eqref{eq:exactdeg} is an equality for the
schedules examined, so that $d_j$ is the exact degree while $N_j$ is a
valid but non-attained upper bound whenever $TE < TR_j$.
\end{corollary}

\begin{proof}
Induction on $j$ using \eqref{eq:recursion}. Each diagonal entry of $D_j$
is $E_{1j} = q_1^{n_j}$ or $E_{2j} = q_2^{n_j}$, and $A_j$ is
independent of $q_1$, $q_2$. Hence, multiplication by $D_jA_j$ increases total degress by at most $n_j$: $\deg(D_jA_j m_{j-1}) \leq n_j + \deg m_{j-1}$. The offset $(1-E_{1j})e$ has degree at most $n_j$. With $\deg m_0 = 0$, induction gives
$\deg m_j \leq \sum_{l=1}^{j} n_l = N_j$, and therefore $\deg s_j \leq N_j$ because $s_j = Pm_j$ with $P$ constant. Equivalently, in
the unrolled form the homogeneous term
$(\prod_{l=j}^{1}D_lA_l)m_0$ has degree at most $\sum_{l=1}^{j} n_l = N_j$
and the term beginning at index $i$ has degree at most $\sum_{l=i}^{j}n_l
\leq N_j$. The upper bound need not be attained because sequence-dependent cancellations or zero coefficients may reduce the actual degree. In the equal-TR case, taking $\Delta=TR$ gives $E_1=q_1, E_2=q_2$, and deg $s_j\leq j$.
The refinement \eqref{eq:exactdeg} follows from the same induction started
at the sampling instant: the $j$-th readout occurs before the relaxation
interval $TR_j$ has elapsed, so the relaxation factors available to $s_j$
are exactly those of $TR_1,\dots,TR_{j-1}$ together with $TE$.
\end{proof}

\begin{remark}
Corollary~\ref{cor:commensurate} converts the exponential-polynomial
representation of Proposition~\ref{prop:expoly} into an ordinary
bivariate polynomial representation whenever the sequence timings are commensurate.
Thus, for any fixed finite MRF schedule, the signal belongs to a finite-dimensional
analytic family even for pseudo-random, non-uniform schedules. The
corollary determines the function class, not the polynomial
coefficients: those coefficients remain sequence dependent and are
controlled by the FA, RF-phase, and TR schedule. Consequently, while results that depend only on analyticity hold broadly, strict local identifiability and global unimodality require additional sequence-dependent conditions. Both the equal-TR and the non-uniform commensurate cases have been checked
numerically: in each the signal is reproduced to machine precision by a
polynomial in $(q_1, q_2)$ of the degree $d_j$ predicted by
\eqref{eq:exactdeg}.
\end{remark}
%========================================================================
\section{Smoothness of the matching objective}
\label{sec:smooth}
My previous heuristic assessment showed the smoothness of the objective function. To provide analytic derivations of the unimodality and other properties of the objective, we need the real-analyticity property. A real-analytic function is considerably more restrictive than a merely
smooth function: locally it is represented exactly by its convergent
Taylor series. Consequently, it cannot possess isolated corners, cusps,
or nontrivial open plateaus.

\begin{theorem}[real analyticity of MRF matching objective]\label{thm:analytic}
On the open set $\Theta^{\circ} = \{\theta \in \Theta : x(\theta) \neq
0\}$, the squared objective $C^2$ and the loss $L$ are real-analytic
functions of $(R_1, R_2)$. The objective $C$ itself is real-analytic at
every point where $C > 0$.
\end{theorem}

\begin{proof}
By Proposition~\ref{prop:expoly} each $s_j$
is entire, hence real-analytic in $(R_1, R_2)$, and so are its real and
imaginary parts. The denominator $\nm{x}^2 = \sum_j \ab{s_j}^2=\sum_j\left[
(\Rea s_j)^2+(\Ima s_j)^2 \right]$ is therefore real-analytic, and it is strictly positive on $\Theta^{\circ}$. The
numerator of $C^2$ is $\ab{z}^2 = (\Rea z)^2 + (\Ima z)^2$ with
$z = \sum_j \overline{u}_{a j}\, s_j(\theta)$ and real-analytic. A
quotient of real-analytic functions with non-vanishing denominator is
real-analytic \cite{Krantz2002}, giving the claim for $C^2$ and hence
for $L = 1 - C^2$. Finally $C = \sqrt{C^2}$ and the square root is
analytic on $(0,\infty)$.
\end{proof}

Theorem~\ref{thm:analytic} formalizes the smoothness observed
numerically in \cite{Wang2018,Wang2020}. Within $\Theta^\circ$, the only
possible loss of smoothness of $C$ occurs where $C=0$; at points where
$x=0$, the normalized fingerprint itself is undefined. In particular,
the true target satisfies $C(\ta)=1$, so $C$ is real analytic in a
neighborhood of the true peak. Analyticity immediately yields two useful
consequences concerning plateaus and stationary regions. It does not,
however, imply that all stationary points are isolated or finite in
number.

\begin{corollary}[MRF objective has no plateaus, no open stationary regions]
\label{cor:noplateau}
Assume $C^2$ is non-constant on the connected component of
$\Theta^{\circ}$ containing the target. Then no level set of $C^2$
contains an open subset. Moreover, its critical set
$\{\nabla C^2 = 0\}$ contains no open subset unless $C^2$ is constant.
\end{corollary}

\begin{proof}
By the identity theorem for real-analytic functions, if a real-analytic
function is constant on a set with nonempty interior, then it is
constant on the connected component of its domain containing that set.
Since $C^2(\ta) = 1$ and $C^2 < 1$ somewhere on the component for a
non-degenerate sequence, $C^2$ is non-constant and no level set contains
an open ball. The critical set is the common zero set of the
real-analytic functions $\partial C^2/\partial R_1$ and
$\partial C^2/\partial R_2$. According to the identity theorem, unless $C^2$ is constant (both derivatives
vanish identically), this common zero set has empty interior. The
critical set may still consist of isolated points or analytic curves.
Thus analyticity excludes flat regions but does not, by itself,
guarantee that all stationary points are isolated or finite in number.
\end{proof}

The absence of open plateaus is practically relevant because a
coarse-to-fine search cannot be guided by objective variation on a
region where the objective is exactly constant. Analyticity also
supports local gradient estimates near an isolated non-degenerate peak.
It does not, however, imply that the full two-dimensional critical set
is finite, and no such claim is needed below.

The next theorem provides the corresponding one-dimensional statement:
along a non-degenerate parameter slice, the exponential-polynomial
structure prevents arbitrarily oscillatory behavior.

\begin{theorem}[MRF objective has finitely many stationary points along a non-degenerate slice]
\label{thm:finitestat}

Fix $T_2$ and regard $C^2$ as a function of $R_1$ alone. Call the slice
non-degenerate if $C^2(\cdot,T_2)$ is not identically constant. Write

\begin{equation}
C^2(R_1,T_2)
=
\frac{P(R_1)}{Q(R_1)},
\label{eq:C2slice}
\end{equation}

where

\begin{equation}
P(R_1)
=
\ab{\ip{\ua}{x(R_1,T_2)}}^2,
\qquad
Q(R_1)
=
\nm{x(R_1,T_2)}^2.
\label{eq:PQslice}
\end{equation}

Then

\begin{equation}
\frac{d}{dR_1}C^2(R_1,T_2)
=
\frac{
P'(R_1)Q(R_1)-P(R_1)Q'(R_1)
}{
Q(R_1)^2
},
\label{eq:C2slice_derivative}
\end{equation}

and the numerator

\begin{equation}
H(R_1)
:=
P'(R_1)Q(R_1)-P(R_1)Q'(R_1)
\label{eq:Hslice}
\end{equation}

is a finite exponential sum in $R_1$. After collecting equal exponents,
write

\begin{equation}
H(R_1)
=
\sum_{r=1}^{M}
d_r \exp(-\lambda_r R_1),
\qquad
d_r\neq 0,
\qquad
\lambda_1<\cdots<\lambda_M.
\label{eq:Hexp}
\end{equation}

Here $M$ is the number of distinct exponents remaining after
collection; it is unrelated to the number of repetitions $N$.
Then $H$ has at most $M-1$ real zeros, counted with multiplicity.
Consequently, $C^2(\cdot,T_2)$ has at most $M-1$ stationary points on
the admissible $R_1$ interval.

When the TRs are commensurate, let

\begin{equation}
q_1=\exp(-\Delta R_1).
\label{eq:q1slice}
\end{equation}

Then $P$ and $Q$ are polynomials in the real variable $q_1$. If

\begin{equation}
\deg P=p,
\qquad
\deg Q=q,
\label{eq:PQdegrees}
\end{equation}

then the stationary points are the zeros, in the admissible $q_1$
interval, of the polynomial

\begin{equation}
\frac{dP}{dq_1}Q
-
P\frac{dQ}{dq_1},
\label{eq:Hpoly}
\end{equation}

whose degree is at most $p+q-1$. Hence
$C^2(\cdot,T_2)$ has at most $p+q-1$ stationary points on the slice.

\end{theorem}

\begin{proof}

By Proposition~\ref{prop:expoly}, after fixing $T_2$, each signal
component can be written as a finite exponential sum

\begin{equation}
s_j(R_1,T_2)
=
\sum_k
c_{jk}\exp(-a_{jk}R_1),
\label{eq:slice_exp_sum}
\end{equation}

where the coefficients $c_{jk}$ may depend on the fixed value of $T_2$
and on the sequence, while the exponents $a_{jk}$ are real and
non-negative. Since $\ua$ is fixed,

\begin{equation}
z(R_1)
:=
\ip{\ua}{x(R_1,T_2)}
=
\sum_j
\overline{u}_{aj}\,
s_j(R_1,T_2)
\label{eq:zslice}
\end{equation}

is also a finite exponential sum in $R_1$.

For real $R_1$, complex conjugation acts only on the coefficients. If

\begin{equation}
z(R_1)
=
\sum_k
c_k\exp(-a_kR_1),
\label{eq:z_general}
\end{equation}

then

\begin{equation}
\ab{z(R_1)}^2
=
z(R_1)\overline{z(R_1)}
=
\sum_{k,l}
c_k\overline{c_l}\,
\exp\bigl(-(a_k+a_l)R_1\bigr).
\label{eq:z_mod_square}
\end{equation}

Thus $P(R_1)=\ab{z(R_1)}^2$ is a finite exponential sum. Similarly,

\begin{equation}
Q(R_1)
=
\nm{x(R_1,T_2)}^2
=
\sum_j
\ab{s_j(R_1,T_2)}^2
\label{eq:Q_exp_sum}
\end{equation}

is a finite exponential sum. Differentiation and multiplication preserve
the class of finite exponential sums, so

\begin{equation}
H(R_1)
=
P'(R_1)Q(R_1)
-
P(R_1)Q'(R_1)
\label{eq:H_again}
\end{equation}

is also a finite exponential sum.

On the admissible slice in $\Theta^\circ$,

\begin{equation}
Q(R_1)
=
\nm{x(R_1,T_2)}^2
>
0.
\label{eq:Q_positive}
\end{equation}

Therefore the quotient rule gives

\begin{equation}
\frac{d}{dR_1}C^2(R_1,T_2)
=
\frac{
P'(R_1)Q(R_1)
-
P(R_1)Q'(R_1)
}{
Q(R_1)^2
}
=
\frac{H(R_1)}{Q(R_1)^2}.
\label{eq:C2_derivative_H}
\end{equation}

Hence the stationary points of $C^2(\cdot,T_2)$ are exactly the zeros
of $H$.

Because the slice is non-degenerate, $C^2(\cdot,T_2)$ is not constant,
and therefore $H$ is not identically zero. After collecting equal
exponents, it can be written as

\begin{equation}
H(R_1)
=
\sum_{r=1}^{M}
d_r\exp(-\lambda_rR_1),
\qquad
d_r\neq0,
\qquad
\lambda_1<\cdots<\lambda_M.
\label{eq:H_chebyshev}
\end{equation}

The functions

\begin{equation}
\exp(-\lambda_1R_1),
\ldots,
\exp(-\lambda_MR_1)
\label{eq:chebyshev_functions}
\end{equation}

form an extended Chebyshev system on any real interval
\cite{Polya1976}. Hence any nontrivial linear combination of M such exponentials has at most
$M-1$ real zeros, counted with multiplicity. It follows that $H$ has at
most $M-1$ zeros and therefore that $C^2(\cdot,T_2)$ has at most
$M-1$ stationary points on the admissible interval.

Now suppose that the TRs are commensurate. By
Corollary~\ref{cor:commensurate}, after fixing $T_2$ each signal
component is a polynomial in

\begin{equation}
q_1
=
\exp(-\Delta R_1).
\label{eq:q1_again}
\end{equation}

Because $R_1$ is real, $q_1$ is real and positive. Thus, if

\begin{equation}
s_j(q_1)
=
\sum_k
c_{jk}q_1^k,
\label{eq:sj_poly}
\end{equation}

then

\begin{equation}
\ab{s_j(q_1)}^2
=
\left(
\sum_k c_{jk}q_1^k
\right)
\left(
\sum_l \overline{c}_{jl}q_1^l
\right),
\label{eq:sj_mod_poly}
\end{equation}

which is an ordinary polynomial in $q_1$. The same argument applies to
$\ab{\ip{\ua}{x(q_1)}}^2$. Hence both $P(q_1)$ and $Q(q_1)$ are
polynomials in $q_1$.

Since

\begin{equation}
\frac{dq_1}{dR_1}
=
-\Delta q_1,
\label{eq:dq1dR1}
\end{equation}

the chain rule and quotient rule give

\begin{equation}
\frac{d}{dR_1}C^2
=
-\Delta q_1
\frac{
\displaystyle
\frac{dP}{dq_1}Q
-
P\frac{dQ}{dq_1}
}{
Q^2
}.
\label{eq:C2qderivative}
\end{equation}

Because $q_1>0$ and $Q>0$, the stationary points are exactly the zeros
of

\begin{equation}
\frac{dP}{dq_1}Q
-
P\frac{dQ}{dq_1}
\label{eq:q_stationary_poly}
\end{equation}

in the admissible $q_1$ interval.

If

\begin{equation}
\deg P=p,
\qquad
\deg Q=q,
\label{eq:degree_PQ_again}
\end{equation}

then

\begin{equation}
\deg\left(
\frac{dP}{dq_1}Q
\right)
\leq
p+q-1,
\qquad
\deg\left(
P\frac{dQ}{dq_1}
\right)
\leq
p+q-1.
\label{eq:degree_terms}
\end{equation}

Therefore

\begin{equation}
\deg\left(
\frac{dP}{dq_1}Q
-
P\frac{dQ}{dq_1}
\right)
\leq
p+q-1.
\label{eq:degree_bound_stationary}
\end{equation}

The polynomial cannot vanish identically on a non-degenerate slice;
otherwise \eqref{eq:C2qderivative} would imply that $C^2$ is constant.
It therefore has at most $p+q-1$ real roots. Finally,
$R_1\mapsto q_1=\exp(-\Delta R_1)$ is one-to-one, so roots in the
admissible $q_1$ interval correspond one-to-one to stationary points in
the admissible $R_1$ interval.

\end{proof}

For any fixed non-degenerate slice of a fixed finite MRF sequence,
Theorem~\ref{thm:finitestat} rules out arbitrarily many isolated
stationary points. It does not exclude a stationary curve in the full
two-dimensional domain. The global analysis therefore relies on
explicit local curvature and finite global certificates rather than on
a stronger two-dimensional finiteness claim.

%========================================================================
\section{The objective is the cosine of a distance}
\label{sec:cosine}

Because $C$ uses the modulus of a normalized inner product, it depends
on $x(\theta)$ only through the complex line spanned by that vector. The
natural state space for the matching problem is therefore complex
projective space $\Proj^{N-1}(\Cpx)$, equipped with the Fubini--Study
metric \cite{Kobayashi1969}, whose geodesic distance is
\begin{equation}\label{eq:fs}
d\bigl( [p], [q] \bigr)
 = \arccos \frac{\ab{\ip{p}{q}}}{\nm{p}\,\nm{q}}
 \in \Bigl[ 0, \frac{\pi}{2} \Bigr].
\end{equation}

\begin{proposition}[metric reformulation]\label{prop:metric}
For every $\theta\in\Theta^\circ$,
\begin{equation}
C(\theta)
=
\cos d\bigl([\ua],[u(\theta)]\bigr),
\qquad
L(\theta)
=
\sin^2 d\bigl([\ua],[u(\theta)]\bigr).
\end{equation}
Consequently, maximizing $C$ over $\Theta^\circ$ is equivalent to
minimizing the Fubini--Study distance from the fixed point $[\ua]$ to
the image of $\Theta^\circ$ under the map
$\theta\mapsto[u(\theta)]$. Because $\cos$ is strictly decreasing on
$[0,\pi/2]$, $C$ is a strictly decreasing function of this projective
distance.
\end{proposition}

Proposition~\ref{prop:metric} reframes the matching problem
geometrically. Correlation decreases monotonically with Fubini--Study
distance in projective signal space; the nontrivial question is whether
distance in parameter space maps monotonically to distance on the
projective signal manifold. This monotonicity can fail if the
parameter-to-signal map $\theta\mapsto[u(\theta)]$ loses local
sensitivity or folds back globally: nearby parameter directions may
become locally indistinguishable, while well-separated parameters may
generate nearly collinear fingerprints. The projective formulation also
supplies the triangle inequality and provides the metric framework for
the Lipschitz bounds used later in the noise analysis and finite global
certificates.

%========================================================================
\section{Local geometry of the peak}
\label{sec:local}

Proposition~\ref{prop:metric} shows that MRF matching is a geometric
nearest-point problem. The natural next question is: what determines the
local shape of the matching peak? To answer this, we introduce the
curvature matrix $G$, evaluated in a neighborhood of the true
parameters. This matrix governs the local curvature of the objective,
the local projective geometry of the signal manifold, and the
statistical precision of parameter estimation.

Since changes in overall signal amplitude or global complex phase do not
affect normalized correlation, only derivative components orthogonal to
the complex line spanned by the fingerprint carry information for
matching. Infinitesimally, moving away from the true fingerprint changes
the projective signal only through these orthogonal components. Their
real inner products define the local metric matrix $G$. Each diagonal element of $G$ measures the sensitivity of the normalized
fingerprint to one relaxation rate after removing changes along the
complex fingerprint line, which correspond to overall amplitude and
global phase. The off-diagonal elements measure coupling between the
$R_1$ and $R_2$ encoding directions. If the two projected derivative
directions become nearly linearly dependent over the real parameter
space, $G$ becomes ill-conditioned and the two parameters become
difficult to distinguish. Based on $G$, we have:

\begin{theorem}[second-order expansion and profiled Fisher information]
\label{thm:fisher}
Let
\begin{equation}
g_k(\theta)
:=
\Pi_u \frac{\partial u}{\partial R_k},
\qquad
\Pi_u := I-u u^{\hrm},
\qquad
k\in\{1,2\},
\end{equation}
where $I$ denotes the identity matrix. Define the $2\times2$ real
symmetric Gram matrix
\begin{equation}
\label{eq:G}
G_{k\ell}(\theta)
=
\Rea
\left\langle
\Pi_u\frac{\partial u}{\partial R_k},
\Pi_u\frac{\partial u}{\partial R_\ell}
\right\rangle,
\qquad
k,\ell\in\{1,2\}.
\end{equation}
Equivalently,
\begin{equation}
G(\theta)
=
\begin{pmatrix}
\nm{g_1}^2 &
\Rea\ip{g_1}{g_2}
\\[2pt]
\Rea\ip{g_2}{g_1} &
\nm{g_2}^2
\end{pmatrix}.
\end{equation}

Let
\begin{equation}
\Delta
=
R(\theta)-R(\ta).
\end{equation}
Then, as $\Delta\to0$,
\begin{equation}
\label{eq:expansion}
C(\theta)^2
=
1-\Delta^{\trn}G(\ta)\Delta
+
O(\nm{\Delta}^3).
\end{equation}
Consequently,
\begin{equation}
\label{eq:distance_squared}
d\bigl([\ua],[u(\theta)]\bigr)^2
=
\Delta^{\trn}G(\ta)\Delta
+
O(\nm{\Delta}^3).
\end{equation}
If $G(\ta)$ is positive definite, then equivalently
\begin{equation}
\label{eq:distance_expansion}
d\bigl([\ua],[u(\theta)]\bigr)
=
\nm{\Delta}_{G(\ta)}
+
O(\nm{\Delta}^2),
\end{equation}
where
\begin{equation}
\nm{\Delta}_{G(\ta)}
:=
\sqrt{\Delta^{\trn}G(\ta)\Delta}.
\end{equation}

Moreover, consider the observation model
\begin{equation}
y=\rho\,x(\theta)+n,
\qquad
n\sim\mathcal{CN}(0,\sigma^2 I),
\end{equation}
with unknown complex nuisance parameter $\rho$. After eliminating
$\rho$, the profiled Fisher information matrix for $(R_1,R_2)$ is
\begin{equation}
\label{eq:fisher_profiled}
\mathcal{I}_{\mathrm{prof}}(\theta)
=
\frac{2S(\theta)^2}{\sigma^2}\,G(\theta),
\qquad
S(\theta):=\ab{\rho}\,\nm{x(\theta)}.
\end{equation}
Hence, whenever $G(\theta)$ is positive definite, the corresponding
Cram\'{e}r--Rao bound is
\begin{equation}
\label{eq:crb}
\operatorname{Cov}(\widehat{R})
\succeq
\frac{\sigma^2}{2S(\theta)^2}\,
G(\theta)^{-1}.
\end{equation}
\end{theorem}

Theorem~\ref{thm:fisher} describes the local peak geometry through three
equivalent quantities. At the true parameter,
\begin{equation}
\label{eq:hessian_G}
\nabla_R^2 C^2(\ta)
=
-2G(\ta).
\end{equation}
Thus, when $G(\ta)$ is positive definite, or equivalently the smallest eigenvalue
$\lmin(G(\ta))>0$, we have
\begin{equation}
\nabla_R^2 C^2(\ta)
\preceq
-2\lmin(G(\ta))\,I,
\end{equation}
where $I$ denotes the $2\times2$ identity matrix. Hence the true
parameter is a strict non-degenerate local maximizer, $R_1$ and $R_2$
are locally identifiable from the normalized fingerprint to first
order, and the local Cram\'{e}r--Rao bound is finite. This equivalence is
strictly local: positive profiled Fisher information at $\ta$ cannot
exclude a distant alias or secondary peak. Conversely, failure of
positive definiteness reflects a local encoding degeneracy rather than
a peculiarity of the correlation objective.

The eigenvalues and eigenvectors of $G$ describe the principal local
directions of the peak. Large eigenvalues correspond to narrow,
strongly encoded directions, whereas small eigenvalues correspond to
broad, weakly encoded directions. The eigenvectors define these
principal directions, while the off-diagonal entry reflects coupling
between the two relaxation-rate sensitivities. We use $G$ because it
simultaneously represents the local projective metric, the curvature of
the matching peak, and, up to the signal-to-noise scale factor in
\eqref{eq:fisher_profiled}, the profiled Fisher information.

Theorem~\ref{thm:fisher} establishes a strict non-degenerate local
maximum when $G(\ta)$ is positive definite, but it does not determine
how far from the optimum the Hessian remains negative definite. To
address this question, we bound the variation of the Hessian through
the third derivative of $C^2$.

\begin{corollary}[explicit radius of strict concavity and radial
monotonicity]
\label{cor:radius}
Suppose
\begin{equation}
\mu
:=
\lmin(G(\ta))
>
0.
\end{equation}
Let $K$ be a compact convex neighborhood of $R(\ta)$ contained in
$\Theta^\circ$, and define
\begin{equation}
K_3
:=
\sup_{R\in K}
\nm{D^3 C^2(R)}.
\end{equation}
With the convention $2\mu/K_3=+\infty$ when $K_3=0$, put
\begin{equation}
\label{eq:rhostar}
\rho^{*}
:=
\min
\left\{
\dis\bigl(R(\ta),\partial K\bigr),
\frac{2\mu}{K_3}
\right\}.
\end{equation}
Then $C^2$ is strictly concave on
\begin{equation}
B
=
\left\{
R:
\nm{R-R(\ta)}<\rho^{*}
\right\}.
\end{equation}
Consequently, $\ta$ is the unique stationary point and the unique
maximizer of $C$ in $B$, and there is no other local maximum in $B$.
Moreover, for every unit vector $e$,
\begin{equation}
t
\longmapsto
C\bigl(R(\ta)+te\bigr)
\end{equation}
is strictly decreasing for $0<t<\rho^{*}$ whenever the corresponding
ray remains in $B$.

Furthermore, $K_3$ is finite, and an explicit computable upper bound can
be obtained from the coefficients and exponents of the finite
exponential sums in Proposition~\ref{prop:expoly}, together with a
positive lower bound on $\nm{x}^2$ over $K$.
\end{corollary}

Corollary~\ref{cor:radius} establishes the local part of the single-peak
property: within an explicitly computable neighborhood, the target is
the unique stationary point and the objective decreases strictly along
every ray. Global concavity is neither required nor expected because the
bounded correlation objective can become relatively flat at low
similarity far from the peak. Likewise, the approximate $T_1/T_2$
orthogonality reported as Property~(III) in \cite{Wang2018} is not
required. Off-diagonal elements of $G$ rotate the principal axes of the
local peak but do not invalidate strict local concavity as long as
$G(\ta)$ remains positive definite. In practice this coupling is not
negligible, which helps explain why interleaved rather than single-pass
coordinate updates are preferable.

%========================================================================
\section{Global uniqueness and radial monotonicity}
\label{sec:global}

Local curvature alone cannot exclude a competing peak far from the
truth, and no such guarantee can hold for every possible acquisition
schedule. We therefore distinguish two global properties. The first is
uniqueness of the global maximizer. The second, stronger property is the
behavior observed in \cite{Wang2018,Wang2020}: strict decrease along
every admissible ray from the true parameters, which also excludes
secondary stationary points. The first can be certified from objective
values; the second requires control of the radial derivative.

\begin{theorem}[analytic structure of the ambiguity set]
\label{thm:ambiguity}
Define the ambiguity set
\begin{equation}
A(\ta)
=
\left\{
\theta\in\Theta^\circ:
C(\theta)=1
\right\}.
\end{equation}
Then $A(\ta)$ is the common zero set, relative to $\Theta^\circ$, of the
complex-valued real-analytic functions
\begin{equation}
\label{eq:minors}
m_{ij}(\theta)
=
s_i(\theta)\,s_j(\ta)
-
s_j(\theta)\,s_i(\ta),
\qquad
1\le i<j\le N.
\end{equation}
Equivalently,
\begin{equation}
A(\ta)
=
\left\{
\theta\in\Theta^\circ:
\Phi(\theta)=0
\right\},
\qquad
\Phi(\theta)
:=
\sum_{i<j}\ab{m_{ij}(\theta)}^2,
\end{equation}
where $\Phi$ is a real-valued real-analytic function. If $G(\ta)$ is
positive definite, then $A(\ta)$ has empty interior in $\Theta^\circ$
and Lebesgue measure zero.

If, in addition, the TRs are commensurate and each signal component has
polynomial degree at most $M_N$ in $(q_1,q_2)$, then each $m_{ij}$ is a
polynomial in $(q_1,q_2)$ of degree at most $M_N$. Hence $A(\ta)$ is a
real algebraic set relative to the admissible parameter domain. If any
two nonzero minors are coprime, their common isolated zeros number at
most $M_N^2$ over the complexification, counted with multiplicity, and
therefore the number of real isolated ambiguity points is also at most
$M_N^2$.
\end{theorem}

Theorem~\ref{thm:ambiguity} addresses exact aliases, but uniqueness of
the global maximizer does not by itself exclude a lower secondary peak.
The next theorem therefore provides two finite certificates: a
value-based certificate for global uniqueness and a radial-derivative
certificate for the stronger global unimodality statement.

\begin{theorem}[finite certificates for global uniqueness and global
radial monotonicity]
\label{thm:certificates}
Let
\begin{equation}
F(R)
:=
C(R)^2,
\qquad
\Ra
:=
R(\ta),
\end{equation}
and suppose that the compact rate-coordinate parameter domain
$\mathcal{D}$ is convex and contained in $\Theta^\circ$. Let
\begin{equation}
B
=
\left\{
R\in\mathcal{D}:
\nm{R-\Ra}<\rho^{*}
\right\}
\end{equation}
be the strictly concave ball from
Corollary~\ref{cor:radius}, and let
\begin{equation}
A
=
\left\{
R\in\mathcal{D}:
\nm{R-\Ra}\geq\rho^{*}
\right\}
\end{equation}
denote the remaining compact parameter region outside $B$.

Define
\begin{equation}
L_C
:=
\sup_{R\in\mathcal{D}}
\sqrt{\lmax(G(R))}.
\end{equation}
Then, for any $R,R'\in\mathcal{D}$,
\begin{equation}
\label{eq:lipschitz}
\ab{C(R)-C(R')}
\leq
d\bigl([u(R)],[u(R')]\bigr)
\leq
L_C\nm{R-R'}.
\end{equation}

For $R\neq\Ra$, define the outward radial unit vector and radial
derivative by
\begin{equation}
e(R)
:=
\frac{R-\Ra}{\nm{R-\Ra}},
\qquad
q(R)
:=
e(R)^{\trn}\nabla F(R).
\end{equation}
On $A$, $q$ is continuously differentiable. Let $H$ be any certified
upper bound satisfying
\begin{equation}
\nm{\nabla q(R)}
\leq H,
\qquad
R\in A.
\end{equation}
One convenient choice is any $H$ no smaller than
\begin{equation}
\label{eq:Hbound}
\sup_{R\in A}
\left[
\nm{\nabla^2F(R)}
+
\frac{\nm{\nabla F(R)}}{\nm{R-\Ra}}
\right].
\end{equation}

\emph{Value certificate.}
Let $\mathcal{G}_A\subset A$ be a finite grid of fill radius $h$ in $A$,
meaning that for every $R\in A$ there exists
$R_g\in\mathcal{G}_A$ such that
\begin{equation}
\nm{R-R_g}\leq h.
\end{equation}
Let
\begin{equation}
V
:=
\max_{R_g\in\mathcal{G}_A} C(R_g).
\end{equation}
If
\begin{equation}
\label{eq:valuecert}
V+L_C h<1,
\end{equation}
then no global maximizer lies outside $B$. Consequently, $\ta$ is the
unique global maximizer of $C$ over $\mathcal{D}$.

\emph{Radial certificate.}
Using a finite grid $\mathcal{G}_A\subset A$ of fill radius $h$, suppose
\begin{equation}
\max_{R_g\in\mathcal{G}_A}q(R_g)
=
-m
<
0.
\end{equation}
If
\begin{equation}
\label{eq:radialcert}
Hh<m,
\end{equation}
then
\begin{equation}
q(R)<0
\end{equation}
everywhere in $A$. Combined with
Corollary~\ref{cor:radius}, $F$ and hence $C$ decrease strictly along
every admissible ray from $\Ra$ throughout $\mathcal{D}$.
Consequently, $\ta$ is the only stationary point and the only local or
global maximizer in $\mathcal{D}$.
\end{theorem}

The two certificates address different practical questions. The value
certificate is sufficient to establish that the target is the unique
global optimum once the local basin is certified; a coarse MRF-ZOOM
sweep supplies the required samples, while the Lipschitz term converts
finite sampling into a continuous-domain guarantee. The
radial-derivative certificate is stronger because it formalizes the full
single-peak, monotone-decay behavior observed in
\cite{Wang2018,Wang2020} and excludes lower secondary extrema. Both
certificates are finite and sequence specific. Neither implies that
every conceivable MRF schedule is globally unimodal.

%========================================================================
\section{Discussion}
\label{sec:discussion}
The uniqueness and convergence of MRF dictionary searching have not yet
been fully established, and this work is a first step in that direction.
For simplicity we treat only $T_1$ and $T_2$ and an uncompressed
dictionary; compressed dictionary searching will be covered in a separate
manuscript. This preprint covers most of the derivations but more proofs and analysis results are needed and will be provided to complete the work.

\subsection{Relation to the three empirical properties}

The results above give a more precise account of the three empirical
properties reported in \cite{Wang2018}. Property (II) decomposes into
unconditional smoothness, a locally strict peak under positive Fisher
information, and sequence-specific global certificates. Property (I)
concerns off-resonance and lies outside the scope of this condensed
treatment. Property
(III), approximate $T_1/T_2$ orthogonality, is not required for
unimodality; the off-diagonal curvature is in general not negligible,
which primarily affects the efficiency of coordinate-wise search rather
than the existence of a strict local peak.

\subsection{Implications for MRF-ZOOM and for sequence design}

Several practical consequences follow (TBD).

\subsection{Limitations}

Several limitations define the scope of the present theory. The analysis
assumes a single tissue compartment parameterized by $(T_1, T_2)$;
off-resonance is not treated here.
Partial-volume mixtures, magnetization transfer, diffusion, flow,
exchange, and other additional physics can enlarge or alter the signal
manifold and may introduce genuine non-identifiability. Sequence
imperfections such as $B_1$ inhomogeneity can be incorporated as fixed
nuisance effects in the forward model, but estimating them as additional
parameters enlarges the information matrix and may degrade conditioning.
The finite global certificates are also sequence- and target-specific;
establishing a uniform guarantee over an entire parameter family
requires uniform bounds on the local curvature, derivative Lipschitz
constants, and global separation margin, and must be evaluated afresh for
any new sequence to which the framework is applied.

\section*{Acknowledgement}
This work was initiated at the time when MRF-ZOOM was eventually published after being rejected  (even as an abstract to ISMRM) more than 8 times within 4-5 years. I would thank one previous reviewer for the comment of lacking an analytic proof though the MRF itself has yet been analytically proven to converge. As my research interest has been shifted away from MRF due to the significantly prolonged publication process, I have not tracked the progress of MRF literature, therefore, the claims for other parameter matching through MRF should be re-examined. Nevertheless, the overall framework can be still useful.


\begin{thebibliography}{99}

\bibitem{Ma2013}
D.~Ma, V.~Gulani, N.~Seiberlich, K.~Liu, J.~L. Sunshine, J.~L. Duerk,
and M.~A. Griswold, ``Magnetic resonance fingerprinting,''
\emph{Nature}, vol.~495, pp.~187--192, 2013.

\bibitem{Jiang2015}
Y.~Jiang, D.~Ma, N.~Seiberlich, V.~Gulani, and M.~A. Griswold,
``MR fingerprinting using fast imaging with steady state precession
(FISP) with spiral readout,'' \emph{Magn. Reson. Med.}, vol.~74,
no.~6, pp.~1621--1631, 2015.

\bibitem{Wang2018}
Z.~Wang, J.~Zhang, D.~Cui, J.~Xie, M.~Lyu, E.~S. Hui, and E.~X. Wu,
``Magnetic resonance fingerprinting using a fast dictionary searching
algorithm: MRF-ZOOM,'' \emph{IEEE Trans. Biomed. Eng.}, vol.~66, no.~6, pp.~1526-1535, 2018,
doi: 10.1109/TBME.2018.2874992.

\bibitem{Wang2020}
Z.~Wang, D.~Cui, J.~Zhang, E.~X. Wu, and E.~S. Hui, ``MRF-ZOOM for the
unbalanced steady-state free precession (ubSSFP) magnetic resonance
fingerprinting,'' \emph{Magn. Reson. Imaging}, vol.~65, pp.~146-154, 2020.

\bibitem{Bloch1946}
F.~Bloch, ``Nuclear induction,'' \emph{Physical Review}, vol.~70, no.~7-8,
pp.~460--474, 1946.

\bibitem{Weigel2015}
M.~Weigel, ``Extended phase graphs: dephasing, RF pulses, and echoes
--- pure and simple,'' \emph{J. Magn. Reson. Imaging}, vol.~41,
no.~2, pp.~266--295, 2015.

\bibitem{McGivney2014}
D.~F. McGivney, E.~Pierre, D.~Ma, Y.~Jiang, H.~Saybasili, V.~Gulani,
and M.~A. Griswold, ``SVD compression for magnetic resonance
fingerprinting in the time domain,'' \emph{IEEE Trans. Med. Imaging},
vol.~33, no.~12, pp.~2311--2322, 2014.

\bibitem{Polya1976}
G.~P\'{o}lya and G.~Szeg\H{o}, \emph{Problems and Theorems in Analysis}. Berlin: Springer, 1976.

\bibitem{Krantz2002}
S.~Krantz and H.~Parks, \emph{A Primer of Real Analytic Functions},
2nd~ed. Boston: Birkh\"{a}user, 2002.

\bibitem{Kobayashi1969}
S.~Kobayashi and K.~Nomizu, \emph{Foundations of Differential Geometry,
Vol.~II}. New York: Interscience, 1996.

\bibitem{Kay1993}
S.~M. Kay, \emph{Fundamentals of Statistical Signal Processing:
Estimation Theory}. Englewood Cliffs, NJ: Prentice Hall, 1993.

\bibitem{Zhao2018}
B.~Zhao, J.~P. Haldar, C.~Liao, D.~Ma, Y.~Jiang, M.~A. Griswold,
K.~Setsompop, and L.~L. Wald, ``Optimal experiment design for magnetic
resonance fingerprinting: Cram\'{e}r--Rao bound meets spin dynamics,''
\emph{IEEE Trans. Med. Imaging}, vol.~38, no.~3, pp.~844--861, 2018.

\bibitem{Asslander2018}
J.~Assl\"{a}nder, M.~A. Cloos, F.~Knoll, D.~K. Sodickson, J.~Hennig,
and R.~Lattanzi, ``Low rank alternating direction method of multipliers
reconstruction for MR fingerprinting,'' \emph{Magn. Reson. Med.},
vol.~79, no.~1, pp.~83--96, 2018.


\bibitem{PaperII}
Z.~Wang, ``The geometry of subspace-compressed magnetic resonance
fingerprinting: structural invariance, information loss, and the choice
of subspace dimension,'' companion manuscript.

\end{thebibliography}
\end{document}